\documentclass[11pt]{article}
\usepackage{amsmath, natbib, hyperref}
\usepackage{graphicx}
\usepackage[english]{babel}
\usepackage[dvipsnames,svgnames]{xcolor}
\newtheorem{theorem}{Theorem}[section]
\newtheorem{lemma}[theorem]{Lemma}

\newenvironment{proof}[1][Proof]{\begin{trivlist}
\item[\hskip \labelsep {\bfseries #1}]}{\end{trivlist}}
\newenvironment{definition}[1][Definition]{\begin{trivlist}
\item[\hskip \labelsep {\bfseries #1}]}{\end{trivlist}}

\newcommand{\qed}{\nobreak \ifvmode \relax \else
      \ifdim\lastskip<1.5em \hskip-\lastskip
      \hskip1.5em plus0em minus0.5em \fi \nobreak
      \vrule height0.75em width0.5em depth0.25em\fi}

\begin{document}

\title{Visualising stock flow consistent models as directed acyclic graphs}
\author{P. Fennell \and D. O'Sullivan \and A. Godin
 \and S. Kinsella\thanks{University of Limerick. Fennell and O'Sullivan: Department of Mathematics and Statistics. Godin and Kinsella: Department of Economics. Corresponding author: \textsf{stephen.kinsella@ul.ie}. The authors gratefully acknowledge the support of the Institute for New Economic Thinking. Fennell and O'Sullivan gratefully acknowledge Science Foundation Ireland grant 11/PI/1026.}}
\maketitle

\begin{abstract}
\noindent We show how every stock-flow consistent model of the macroeconomy can be represented as a directed acyclic graph. The advantages of representing the model in this way include graphical clarity, causal inference, and model specification. We provide many examples implemented with a new software package. 
\end{abstract}


\section{Introduction}
\label{sec:intro}
%SK:This just stops any conflicts with us editing the same files again. The \input command just allows the text to be read in from another file. UPDATE. Turns out you can't bibtex easily so I've switched back.
%\input{sec_intro.tex}

% SK Still room to cut & simplify/sexify here.

%what we do.
In this paper we rigorously show that for every stock flow consistent macroeconomic model there is a corresponding directed acyclic graph which is unique. This is achieved by constructing the graphical representation of the stock flow consistent model and using graph theoretical techniques to decompose the graph into an acyclic graph. We illustrate the theory with an example, and provide details of a computational package that gives the directed acyclic graph representation of any stock flow consistent model.

%What SFCs are
A stock flow consistent model of an economy is built from the national flow of funds accounts, treating sectoral interlinkages carefully to ensure consistency of stocks and flows over time \citep{godley:2007}. A set of behavioural equations supplies the causal assumptions and parameterisations required for modeling changes in policy and changes in the business cycle. Stock flow consistent models are rapidly developing as alternatives to more traditional models of the macroeconomy as the recent survey by \cite{caverzasi2014post} shows. These models are complicated structures with typically several hundred identity and accounting equations as well as a handful of behavioural equations. Describing and displaying the balance sheet and flow matrices, in addition to working out the results of the model, can take up a lot of time and space within a paper, which might better be used for exposition of the model's properties.

%Repeating ourselves a bit. 
%This approach applies rigorous accounting identities to ensure that each flow can be effectively traced through double entry style accounting (the requirement that each flow be represented as a liability and assets for corresponding stocks). Further conditions are placed on how each sector how flows reacts to relative changes in stocks and flows, these conditions are used to represent behavioural changes by each sector to the wider dynamics of the economy.In particular, a stock-flow consistent model attempts to encapsulate all the economic transactions across an economy by describing the financial transactions observed in the national accounts as a series of flows between these stocks. 

%There is an open question
An open question in the stock flow consistent (SFC) literature is the specification of the causal structure within the model. Typically\footnote{This form of the consumption function follows \cite{godley:2007}.} one might see a linear consumption function of the form $C = \alpha_{1}YD+\alpha_{2}V_{-1}$, where $YD$ is the flow of disposable income within a given period, $V_{-1}$ is the stock household wealth from previous periods, and $\alpha_{1,2}$ are parameters taking values between zero and one. In this case, causality runs from left to right: when disposable income increases, consumption increases. 

%The problem 
In this paper we develop tools to ask how justified is this, or any other, assumption of causality within SFC models. We argue that one method for gaining a deeper understanding of the causality within any SFC model is through its representation as a Directed Acyclic Graph (DAG). 

%What are  DAGS
Directed graphs have been applied in various fields to establish cause-causality relationships to great effect. In an interacting system, the activation or expression of one variable can case an effect in other variables. This effect is generally unknown as the system is observed as a whole and not in a controlled environment where variables can be isolated and then observed independently of the system. The general aim therefore is reverse engineering, inferring variable interactions from system wide observations \citep{morgan2007counterfactuals}, \citep{pearl2000causality}, \citep{lauritzen2001causal}.

DAG models have been used in the economic literature to model price discovery, \citep{haigh2004causality}, contagion in stock markets, \citep{yang2008contagion}, \citep{bessler2003structure}, the relationship between money and prices\citep{bessler2002money}, \citep{wang2006price}, and price dynamics in agricultural markets \citep{bessler:2003}. In a series of papers including \cite{hoover:2001} and \cite{hoover:2003}, Kevin Hoover has explicitly tied the development of causality in macroeconomic models to DAG representation, while \cite{demiralp2003searching} have explored causal discovery in vector autoregression models. To the best of our knowledge these tools have not been applied to macroeconomic models of the type we study in this paper.

As previously mentioned, we aim to show the correspondence between a stock flow consistent model and a DAG. In a graphical representation of a stock flow consistent model, each node in the graph corresponds to a variable of the model while a directed edge from node $i$ to node $j$ exists if the equation for the variable corresponding to $j$ in the model is a function of the variable corresponding to $i$. For example in the \textsf{SIM} model of \cite{godley:2007}, output $Y$ is a function of Government spending $G$ and household consumption $C$ as given by $Y=G+C$. Thus in the graphical representation of this model, there are directed edges from $G$ to $Y$ and from $C$ to $Y$ (Fig.~\ref{fig:eqn_to_DAG}).
 
\begin{figure}
  \centering
  \includegraphics[width=0.7\textwidth]{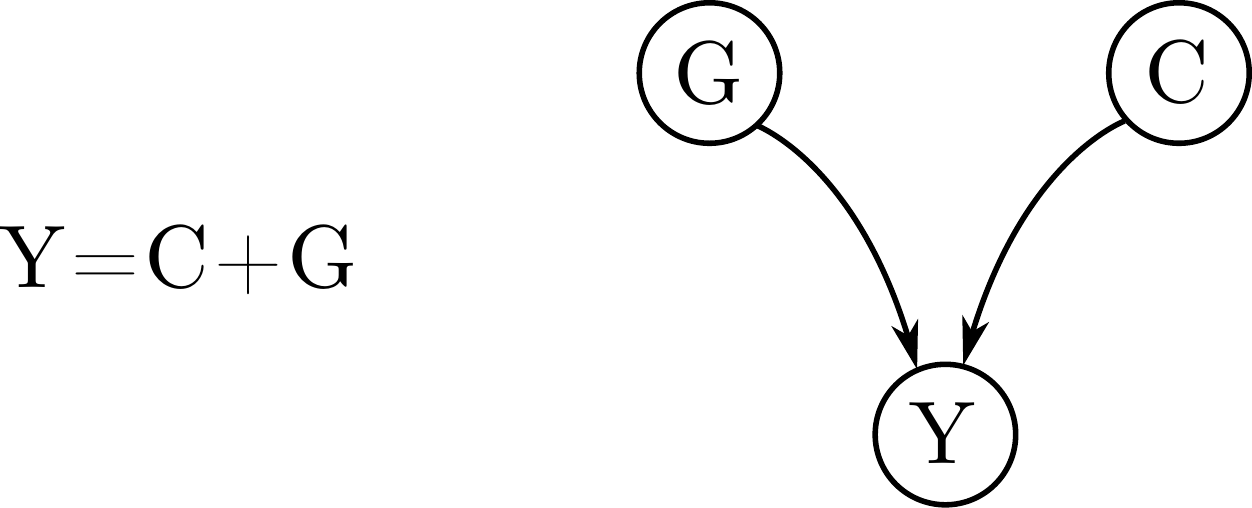}
  \caption{Left: The equation for $Y$ in the SIM model of~\cite{godley:2007}. Right: The graph representation of the equation. Each node corresponds to a variables in the equation. Changes in $C$ and $G$ cause a change in $Y$ and so there are directed links from $C$ to $Y$ and from $G$ to $Y$.}
  \label{fig:eqn_to_DAG}
\end{figure}

This graphical representation of a stock flow consistent model is extremely useful. It allows us to establish the cause and effect mechanisms in the model. The number of equations in such models can be of the order of hundreds~\citep{taylor:2008}, and although the models can be solved by numerical methods is it difficult to analyse the model in terms of the cause-effect relationship between variables. Through the graphical representation provided by DAGs, the effect of exogenous variables - such as government spending -  on the whole system can be examined by following all paths extending from the variable. 

%Last Para
The paper is organized as follows. The equivalence between any stock flow consisitent model and its DAG is derived in Section~\ref{sec:marrying}, while an example is given in Section~\ref{sec:application} of the \textsf{OPEN} model, developed by \cite{godley:2007}. Finally, we conclude and give directions for future research.

\section{Marrying DAGs to SFC Models}
\label{sec:marrying}

\begin{definition}
A \emph{stock flow consistent model} is a macroecenomic model based on both the balance sheet of an economy and a transaction flow matrix $T$ which describes the flow of funds between all sectors of that economy. The model variables $\mathbf{v} = (v_1, v_2, \dots, v_n)$ (ex. $(C, G, I, X,\dots)$) are the non-zero entries of $T$ and the relationships between the variables are described by the system of equations $\mathbf{v}=\mathbf{f}(\mathbf{v})$ i.e.
\begin{equation}
  \mathbf{v} = \begin{pmatrix}
    v_1 \\ v_2 \\ \vdots \\ v_n
  \end{pmatrix} 
  = \begin{pmatrix} 
    f_1(\mathbf{v}) \\
    f_2(\mathbf{v}) \\
    \vdots \\
    f_3(\mathbf{v}) \\
    \end{pmatrix} = \mathbf{f}(\mathbf{v})
\end{equation}
These equations ensure consistency of stocks and flows over time.
\end{definition}

A SFC model can be converted to a directed graph in the following manner. From the set of equations of the SFC we define the Jacobian as
\begin{equation}
  J = \frac{\partial\mathbf{f}}{\partial \mathbf{v}}  =  \begin{pmatrix}
  \frac{\partial f_1}{\partial v_1} & \frac{\partial f_1}{\partial v_2} & \cdots & \frac{\partial f_1}{\partial v_n} \\
  \frac{\partial f_2}{\partial v_1} & \frac{\partial f_2}{\partial v_2} & \cdots & \frac{\partial f_2}{\partial v_n} \\
  \vdots  & \vdots  & \ddots & \vdots  \\
  \frac{\partial f_n}{\partial v_1} & \frac{\partial f_n}{\partial v_2} & \cdots & \frac{\partial f_n}{\partial v_n}
 \end{pmatrix}
\end{equation}
This Jacobian gives us a medium to infer causality in the model through the mathematical operation of differentiation. A variable $v_i$ is directly dependent - or caused by - another variable $v_j$ if and only if 
\begin{equation}
  \frac{\partial v_i}{\partial v_j} \neq 0.
  \label{eq:partialder}
\end{equation}
Since $v_i$ is fully determined by the equation $v_i = f_i(\mathbf{v})$, then Eq.~\eqref{eq:partialder} is equivalent to the condition that $\partial f_i / \partial v_j \neq 0$. Thus $v_j$ causes an effect in $v_i$ if and only if the $(i,j)^{\mbox{th}}$ entry of $J$ is non-zero. These causal dependancies are encoded as a binary matrix $A^{SFC}$  that depends on the SFC model through $J$ by
\begin{equation}
  A^{SFC}_{ij} = \begin{cases}
    1 & \mbox{if }J_{ij} \neq 0 \\
    0 & \mbox{if }J_{ij} = 0 
    \end{cases}
\label{eq:ASFC}
\end{equation} 
This causality will be the foundation of the DAG, but before continuing with the derivation we introduce and formally define the necessary graph theoretical concepts.

\begin{definition}
A \textit{directed graph} $G = (V,E)$ is a set of nodes $V = \left\{v_1,v_2,...,v_n\right\}$ along with a set of directed edges $E = \left\{e_1,e_2,...,e_n\right\}$ that link the nodes. Each edge $e_i$ is of the form $e_i = (v_{i_1},v_{i_2})$ indicating that there is a link from node $v_{i_1}$ to node $v_{i_2}$.  
\end{definition}

A useful representation of a directed network is the \emph{adjacency matrix}. This is a binary matrix $A$ where
\begin{equation}
  A_{ij} = \begin{cases}
    1 & \mbox{if the is a link between $v_i$ and $v_j$} \\
    0 & \mbox{otherwise}
    \end{cases}
\end{equation}
The adjacency matrix is a great mathematical tool used to calculate a range of properties about the network. One such property, which is important in our setting, is the existence of cycles. 

\begin{definition}
In a directed network, a \emph{cycle} is a closed loop of edges where the direction of each edge points the same way around the loop. A directed network that has no cycles is called a \emph{directed acyclic network}, or \emph{DAG}. 
\end{definition}
To find out if a network is acyclic it is sufficient to examine the eigenvalues of adjacency matrix. If all of the eigenvalues of the adjacency matrix are equal to zero then the network is acyclic. Otherwise cycles exist, each of which can be mapped to a strongly connected component.

\begin{definition}
A \emph{strongly connected component} in a directed network $G=(V,E)$ is a maximal subset $V_S \in V$ of nodes such that every pair of nodes in $V_S$ are connected by directed paths in both directions. Maximal here means that no additional nodes in $V$ can be included in $V_S$ without breaking its property of being strongly connected. Every node in the network belongs to one and only one strongly connected component, and strongly connected components of only one node may exist. The set $S = \{S_1, S_2,\dots, S_m\}$ of strongly connected components of $G$ forms a partition of $G$ and this partition is unique.
\end{definition}
In a strongly connected component of more than one node, every node is part of a cycle. This is intuitive, as to be part of the SCC there must be a directed path between the node itself and every other node in the SCC in both directions. An implication of this is that the set of nodes forming any cycle is a subset of exactly one strongly connected component.

\begin{figure}[t]
  \centering
  \includegraphics[width = 0.5\textwidth]{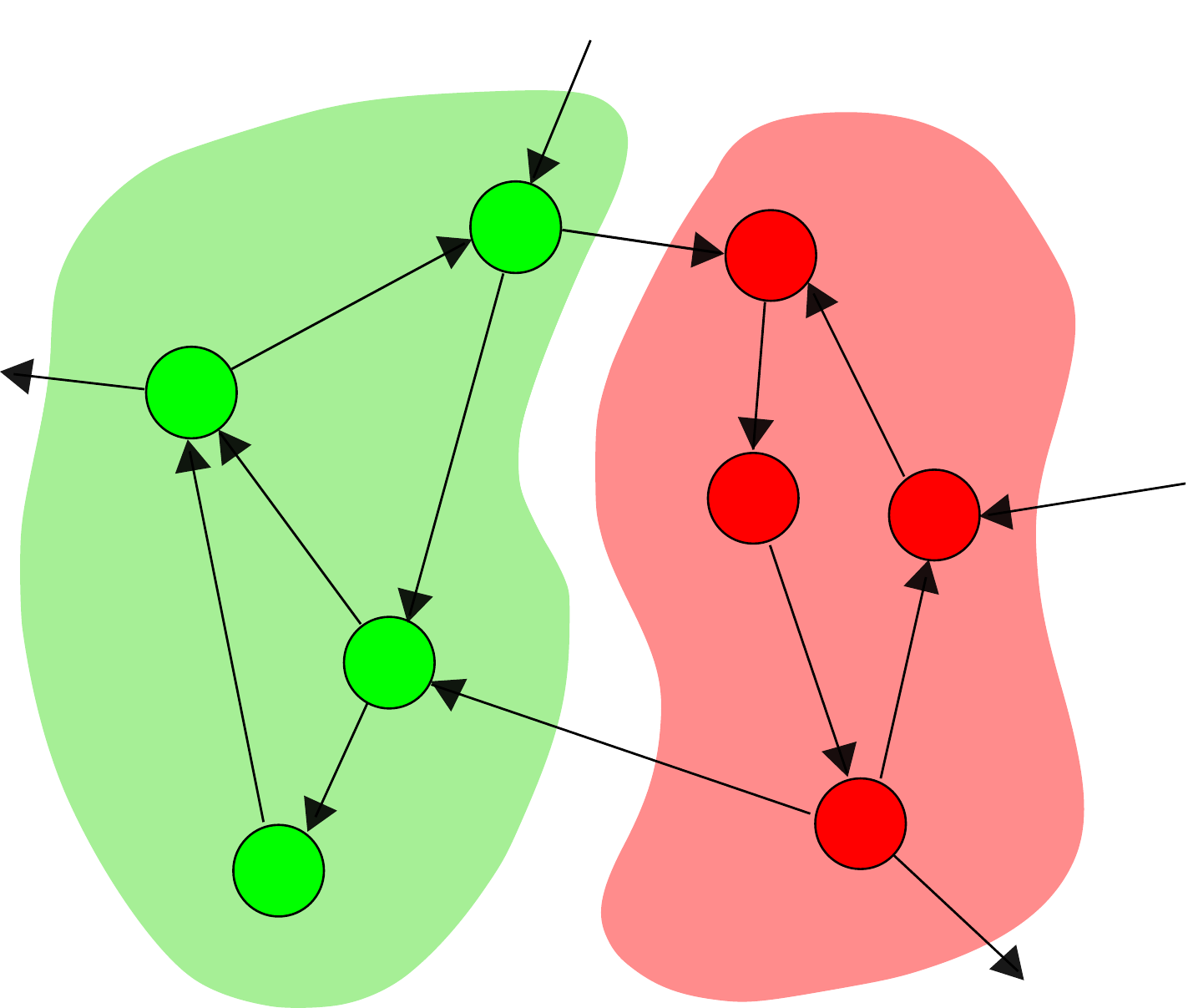}
  \caption{Schematic of the proof of Lemma~\ref{lemma}. In the condensation graph, the stongly connected components $S_1$ (green) and $S_2$ (red) are replaced by metanodes. If these two metanodes form a cycle, then the nodes in $S_1$ and $S_2$ are linked be directed paths in both directions, which contradicts the maximal property of stronlgy connected components.}
  \label{fig:lemma}
\end{figure}

\begin{lemma}
Let $G=(V, E)$ be a directed graph and $S = \{S_1, S_2,\dots, S_m\}$ be the set of $m$ strongly connected components of $G$. The \emph{condensation graph} $G_C$ of $G$ is the graph whose strongly connected componenents are contracted into single nodes called \emph{metanodes}\footnote{Efficient algorithms which identify the strongly connected components exist such as Trajan's algorithm  \cite{tarjan1972depth} or Kosaraju's algorithm \cite{hopcroft1983data}.}. Then the condensation graph is a DAG. 
\label{lemma}
\end{lemma}

\begin{proof}
We aim to show that there are no cycles in $G_C$ in which case it is a DAG. This will be proved by contradiction. A schematic of the proof is given in Fig.~\ref{fig:lemma}.

Assume that there exists a cycle in $G_C$. Then by definition, there is a directed path between at least two vertices in $G_C$ in both directions. Because the nodes in $G_C$ are the stongly connected components in $G$, this implies that there is directed path in both directions between at least two strongly connected components of $G$. However this cannot be the case. The strongly connected components of $G$ are maximal implying that for every node outside the SCC there is not a directed path in both directions to the SCC. Thus two stronlgy connected components cannot be linked by paths in both directions, and so we arrive at a contradiction. 
\end{proof}

The construction of a DAG from the SFC model now follows from the graph theoretical concepts introduced here. Recall that the behaviour of the variables in the SFC is governed by the system of equations $\mathbf{v}=\mathbf{f}(\mathbf{v})$ with Jacobian $J = \partial \mathbf{f} / \partial \mathbf{v}$. The dependencies between variables in the system was encoded in the matrix $A^{SFC}$ which is related to $J$ through Eq.~\eqref{eq:ASFC}. We now construct a directed graph $G$ from the SFC as the graph whose adjacency matrix is $A^{SFC}$. The condensation graph of this directed graph is taken~ and by Lemma~\ref{lemma} this is a DAG. Since there is only one unique partitioning of a graph into its strongly connected components then the DAG is unique.

The implementation of the DAG construction is performed in \textsf{R} and is available in a package from the authors webpage. This package takes the system of equations in the SFC model and returns various outputs. These include the initial directed graph, the strongly connected components which contain all of the cycles and the unique DAG where the cycles have been replaced by metanodes. In the next section we illustrate two examples of the conversion of a SFC into its DAG. We explain why this is important and show the insights into the model that this can afford.

\newpage

\section{Application}
\label{sec:application}

\subsection{The Open Economy Model}
In Chapter 12 of~\cite{godley:2007}, a stock flow consistent model of two open, interacting economies is given. This includes the balance sheet of the two countries along with the transaction flow matrix which describes the flows of assets and liabilites within and between the countries. A set of equations is presented which describes how each of the variables in the blanace sheet and transaction flow matrix are related. This set contains over one hundred equations with many of these equations ensuring that the model is closed and so satisfies double entry accounting standards. Because of the size of the system a rigorous analysis is difficult. The effect of exogeneous paramters and variables, or the depth of their effect, is difficult to ascertain. The DAG provides a highly useful tool to this effect. 

\begin{figure}
  \centering
  \includegraphics[width=0.45\textwidth]{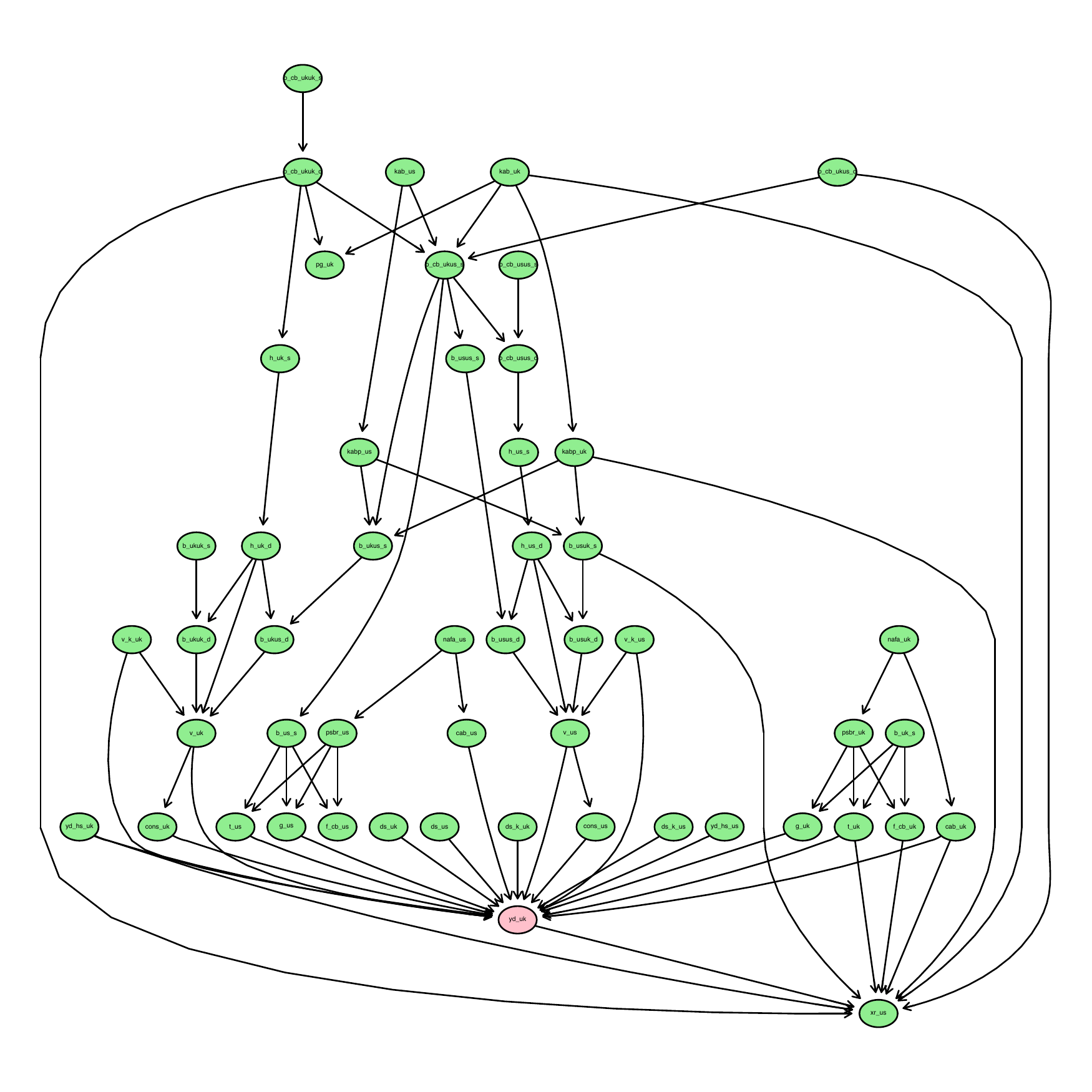}
  \label{fig:DAG_openfix}
  \includegraphics[width=0.45\textwidth]{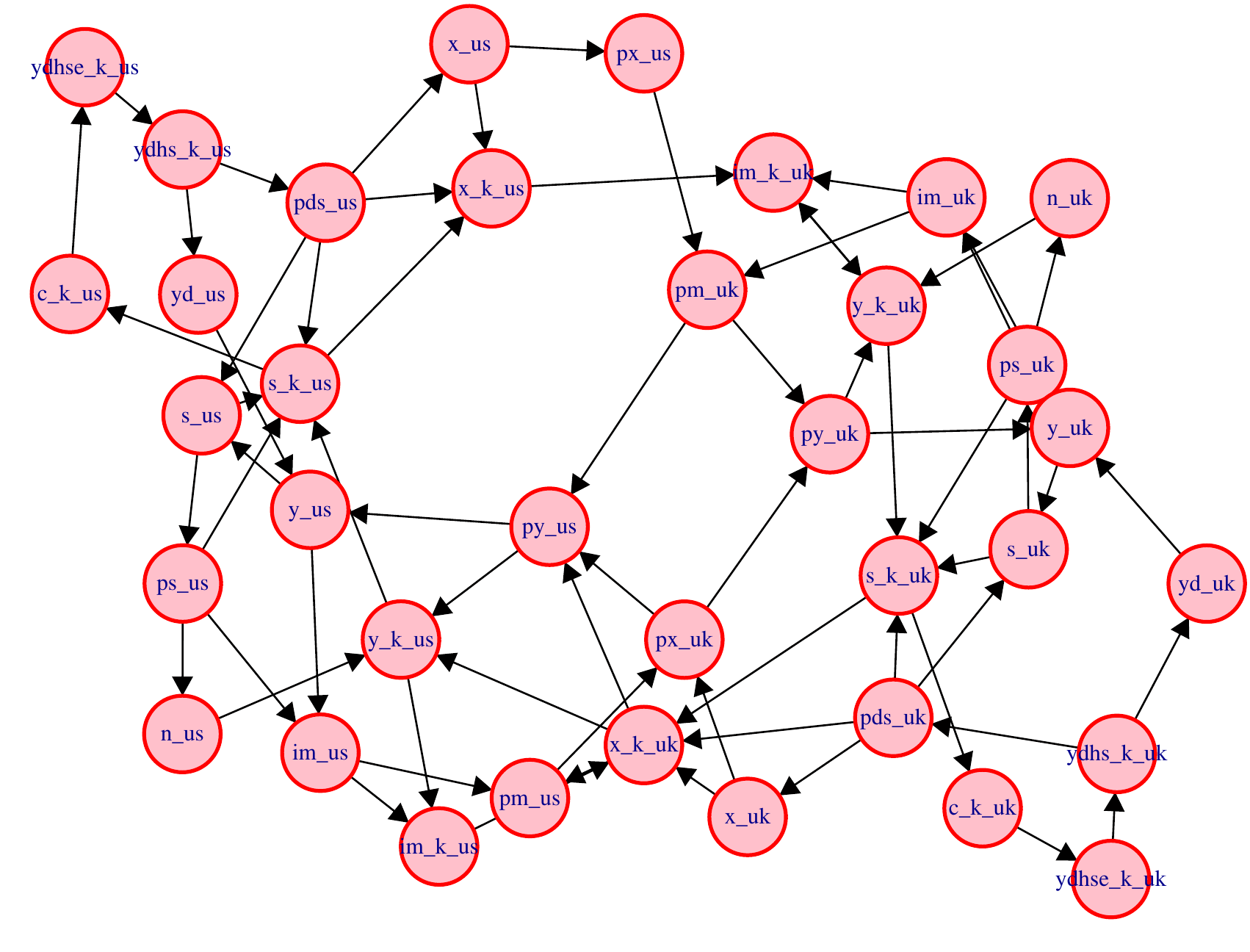}
  \label{fig:SCC_openfix}
  \caption{Left: DAG of the Open Economy SFC model. Each green node corresponds to a variable of the model; these nodes do not belong to any loops. The red node is a metanode in the condensation graph. This corresponds to a strongly connected component in the original directed graph. Right: The stongly connected component which is contracted into a metanode.}
  \label{fig:openfix}
\end{figure}

The DAG of the model is shown in Fig.~\ref{fig:openfix}. The directed graph attained from the Jacobian of the system of equations of the model has one non-trivial strongly connected component. This strongly connected component is shown in Fig.~\ref{fig:openfix}. It can be seen that every node in the SCC can both reach and be reached from every other node in the SCC along some directed path. Once the strongly connected component has been replaced by a metanode - as illustrated with the red node in Fig.~\ref{fig:openfix} - the directed graph becomes a DAG.

\section{Conclusion}
\label{sec:conclusion}

The objective of this paper was to rigorously show that for every stock flow consistent (SFC) macroeconomic model there is a corresponding directed acyclic graph (DAG) which is unique. 

We construct the DAG of any stock flow consistent model in the following manner. Firstly, we use Jacobian methods on the system of equations of the stock flow consistent model to form a directed graph. This directed graph is then decomposed into its corresponding condensation graph by replacing the strongly connected components with single metanodes (Fig.~\ref{fig:lemma}). Since every cycle in the directed graph is contained in exactly one strongly connected component, the condensation graph is acyclic and so we have formed a DAG.

We illustrate the theory with an example, the \textsf{OPEN} model developed by \cite{godley:2007}, and provide details of a computational tool that gives the directed acyclic graphical representation of any stock flow consistent model. We have used this package to generate the DAGs of almost every model within \cite{godley:2007} and these are available online\footnote{The website is \href{http://stephenkinsella.net}{here}.}.

This formal linkage of two rather disparate fields is important. The DAG makes it much easier to visualise large macroeconomic models. This is useful not only for analysis but also for efficiently solving the model computationally. The topological ordering of the system (Fig.~\ref{fig:openfix}) - especially the isolation of cycles - can lead to great improvements in computational speed, which decreases as the model increases in size.

This linkage is important as directed acyclic graphs allow us to use well understood techniques for system-wide causal discovery \citep{pearl2000causality}. Graph-theoretic search methods have not, typically, been used for time series data, the data type we normally use in macroeconomic models.

Our further work will concentrate on recovering directed acyclic graphs from empirical stock flow consistent models to aid in model selection and to generate a topological ordering within the model structures.

\bibliographystyle{Chicago}
\bibliography{References_DAG}

\end{document}